\documentclass[11pt]{article}
\usepackage{graphicx} 
\usepackage{caption}
\usepackage{subcaption}
\title{Towards Practical Robustness Auditing for Linear Regression}
\author{Daniel Freund\\ MIT \and Samuel B. Hopkins \\ MIT}

\usepackage{hopkins}

\begin{document}

\maketitle

\thispagestyle{empty}

\begin{abstract}
    We investigate practical algorithms to find or disprove the existence of small subsets of a dataset which, when removed, reverse the sign of a coefficient in an ordinary least squares regression involving that dataset.
    We empirically study the performance of well-established algorithmic techniques for this task -- mixed integer quadratically constrained optimization for general linear regression problems and exact greedy methods for special cases.
    We show that these methods largely outperform the state of the art \cite{broderick2020automatic,moitra2022provably,kuschnig2021hidden} and provide a useful robustness check for regression problems in a few dimensions.
    However, significant computational bottlenecks remain, especially for the important task of disproving the existence of such small sets of influential samples for regression problems of dimension $3$ or greater.
    We make some headway on this challenge via a spectral algorithm using ideas drawn from recent innovations in algorithmic robust statistics.
    We summarize the limitations of known techniques in several challenge datasets to encourage further algorithmic innovation.
\end{abstract}

\newpage

\thispagestyle{empty}
\tableofcontents

\newpage

\setcounter{page}{1}

\section{Introduction}
Recently, Broderick, Giordano, and Meager \cite{broderick2020automatic} identified a striking pattern of non-robustness in several high-quality and large-scale econometric studies, ranging from the effects of healthcare enrollment (Oregon Medicaid Study \cite{finkelstein2012oregon}) to the effects of microcredit in developing economies \cite{angelucci2015microcredit, augsburg2015impacts,attanasio2015impacts,banerjee2015miracle,crepon2015estimating,karlan2011microcredit,tarozzi2015impacts}.
Key conclusions of the studies -- e.g. that the effect of a treatment on some outcome is positive and statistically significant -- change when the statistical analyses are re-run with a small but carefully-chosen subset of the sample is removed.
Often, dropping less than $1\%$ of the observations, or even just a single one, suffices.
This phenomenon appears even when the authors of the original studies have run careful robustness checks and removed outliers,
and it appears in the simplest settings, such as regression of a single binary treatment variable against a single real-valued outcome \cite{broderick2020automatic}.
Subsequent work by Kuschnig, Zens, and Crespo Cuaresma \cite{kuschnig2021hidden} reinforces this theme.

This leads to significant concerns about generalization and replicability.
Large studies necessarily use data collection methods which deviate from ``truly'' random sampling -- they use surveys and public records, they impute missing data, etc.
Although study designers devote significant effort to correcting for these imperfections e.g., by re-weighing subpopulations, such correction schemes may still yield samples whose distribution deviates by a few percent from ``truly random'' draws from the underlying population.
Furthermore, even the most careful randomized controlled trials may in other ways draw samples from a population which differs from the population to which the conclusions purportedly apply -- for example, a drug trial conducted in Boston whose results will be used to support nation-wide use of a drug, or the impact of a policy trialed in one country that may be used to inform its use in another.

A false discovery arising because of a few-percent deviation between the sample and underlying population would induce a small, highly influential set of samples -- thus, finding such small influential sets is a potential avenue to detecting such false discoveries. 
This is not the only reason a small influential set could arise, however -- another possibility is that the sample is indeed reflective of the population, and the effect detected in the sample is driven by a small segment of the population.
In this case, too, the small influential set (or knowledge that there is none) should be of significant interest to the researcher.
For further statistical interpretation of small influential sets of samples, see
\cite{broderick2020automatic,kuschnig2021hidden,moitra2022provably} for more discussion.

\paragraph{Robustness Auditing}
Assuming that we are interested in knowing whether the conclusion of a study is non-robust to the removal of a few samples, how can we find the offending samples?
Conversely, how can we be sure that such samples do \emph{not} exist?
Following Moitra and Rohatgi \cite{moitra2022provably}, we call this type of algorithmic task \emph{robustness auditing}.
We focus on least-squares linear regression, due to its simplicity and widespread use.

Concretely, we study algorithms to compute or approximate the following quantity, given a dataset $(X,y) = (X_1,y_1),\ldots,(X_n,y_n) \in \R^{d+1}$ and a coordinate $i \in [d]$:
\[
  \text{Stability}(X,y) = \min \, |S| \, \text{ such that } \, \text{sign}(\beta^{OLS, \, [n] \setminus S}_i) \neq \text{sign}(\beta^{OLS, [n]}_i)
\]
where $\beta^{OLS, \, [n]\setminus S}_i$ denotes the $i$-th coordinate of the ordinary least squares regression vector using the dataset $\{(X_i,y_i)\}_{[n] \setminus S}$.
A naive algorithm to decide whether $\text{Stability}(X,y) \leq t$ is to run one instance of linear regression for each $S \subseteq [n]$ with $|S| \leq t$. 
But this is computationally intractable even for moderate values of $n$ and $t$.
Moitra and Rohatgi \cite{moitra2022provably} show that this intractability is, to some extent, inherent: under standard computational complexity hypotheses, any algorithm that provably computes $\text{Stability}(X,y)$\footnote{Moitra and Rohatgi actually study a \emph{fractional} variant of Stability, where samples can be re-weighted instead of completely removed.} for any $d$-dimensional $X,y$ has worst-case running time at least $n^{\Omega(d)}$.

However, worst-case computational intractability does not preclude the existence of algorithms which compute or approximate Stability well for many instances of linear regression encountered in practice.
In this work we design several algorithms which can perform two related tasks: 
\begin{enumerate}
    \item produce a small set $S$, thus providing an \emph{upper bound} on $\text{Stability}(X,y)$, and
    \item provide a \emph{lower bound} on $\text{Stability}(X,y)$ which is valid on a per-dataset basis -- i.e. which holds without any statistical assumptions on $X,y$.
\end{enumerate}
We highlight that lower bounds on Stability are crucial for robustness auditing.
Optimization methods that approximate $\text{Stability}(X,y)$ by finding some subset of samples $S$ whose removal changes $\text{sign}(\beta_i)$ with no guarantee that it is the smallest one 
provide only upper bounds.
This is because given such a subset, there is no way to know if an even smaller set might exist, meaning that algorithms which provide only upper bounds on Stability have limited utility for checking robustness.

\paragraph{Our Contributions}
We design and implement several algorithms for robustness auditing and study their performance empirically on a range of regression problems drawn from recent publications in economics and political science, as well as standard testbed datasets such as Boston Housing \cite{harrison1978hedonic}.
We focus largely on simple and well-studied algorithmic ideas, to lay a foundation for future innovation by establishing performance baselines.
We provide implementations of these algorithms in an accompanying Python software package, \texttt{auditor\_tools}.\footnote{\url{https://github.com/df365/robustness_auditing/tree/main}}

Our main thesis is that standard algorithmic techniques actually provide significantly better than state-of-the-art performance for robustness auditing.
Moreover, the performance of some of these approaches is good enough for mainstream adoption, at least in simple-enough regression problems.
In more detail, we study:
\begin{itemize}
    \item An mixed integer quadratically constrained optimization approach to approximate (or sometimes exactly compute) Stability, using off-the-shelf optimization software (Gurobi \cite{achterberg2020s}). 
    On laptop hardware, this approach scales to regression problems with $10-100$ dimensions and $\approx 10^4$ samples, or $100-1000$ dimensions and $\approx 1000$ samples.
    It provides matching upper and lower bounds on Stability for numerous regression problems drawn from recent publications and standard testbed datasets where prior methods either provide no lower bound or have large gaps between upper and lower bounds.

    \item Efficient greedy algorithms to compute Stability exactly in two simple but common special cases of linear regression: regression of a real-valued outcome against a binary treatment variable and difference-in-difference estimation.
    These algorithms only apply to these special cases but always provide matching upper and lower bounds on Stability and easily scale to datasets with millions of samples.
\end{itemize}

Using these algorithms, we audit several linear regression datasets drawn from prominent econometric studies which prior work investigated using the algorithm of \cite{broderick2020automatic}, ZAMinfluence. We find in these examples that ZAMinfluence frequently fails to find the smallest possible set of samples $S$ to change the sign of an OLS regression coordinate. (Moitra and Rohatgi observe a similar phenomenon in the Boston Housing dataset \cite{moitra2022provably}.)

In some cases \cite{martinez2022much,eubank2022enfranchisement}, published in leading journals in economics and political science, the authors report the number of samples returned by ZAMinfluence as the minimum needed to change the outcome of their study, as evidence of robustness of their results.
That is, they treat the ZAMinfluence \emph{upper bound} on Stability as if it were a \emph{lower bound}, although ZAMinfluence comes with no such guarantee.
We invalidate such claims by finding smaller subsets, thus  highlighting the importance of lower bounds on Stability.
On datasets of small dimension (two or three), our algorithms frequently provide matching lower and upper bounds.

Overall, our results suggest that greedy methods and mixed integer quadratically constrained optimization offer a useful approach to robustness auditing for many regression problems encountered in practice.
However, there is room for improvement, especially on regression problems with more than two or three dimensions: on many such datasets drawn from econometric studies, none of our algorithms, or those in prior work, provide any nontrivial lower bounds on stability (in a reasonable amount of computation time).

In fact, it is easy to construct simple synthetic datasets with $100$ samples in four or more dimensions with Gaussian covariates for which no prior algorithm, nor any of the above, offer any nontrivial lower bound on Stability (in a reasonable amount of computation time).
Following recent theoretical developments in algorithmic robust statistics \cite{klivans2018efficient,bakshi2021robust}, our third contribution shows empirically that this is not due to inherent computational intractability.
We implement:
\begin{itemize}
    \item A spectral algorithm (i.e, based on eigenvalues and eigenvectors) which gives nontrivial lower bounds on Stability for synthetic datasets with tens of thousands of samples and four (or more) dimensions. In these settings, we did not plant any outliers and consequently the resulting regressions are expected to be stable under the removal of a sizable fraction of the samples.  Nonetheless, the mixed integer optimization approach provides no lower bound greater 0,  whereas \cite{moitra2022provably} does not run (in reasonable time) even in dimension 4.
    This shows that there is room for improvement beyond baseline approaches for robustness auditing. However, as our spectral algorithm is heavily tailored to synethic datasets, it does not improve over our baseline methods on any of the real-world data we study.
\end{itemize}

To encourage future work, we summarize the limitations of the algorithms we study in several ``challenge datasets,''.
These are datasets where our algorithms and those of prior work leave large gaps between upper and lower bounds on Stability; new algorithms which shrink or close these gaps would thus represent progress on robustness auditing for linear regression.

\subsection{Related Work}
We are aware of three prior works which attempt to compute or approximate $\text{Stability}(X,y)$.

\paragraph{ZAMinfluence and refinements} ZAMinfluence \cite{broderick2020automatic} 
finds small subsets of samples to drop based on the classical notion of influence functions.\footnote{Essentially, it removes those samples which have the greatest effect on the fitted parameter of interest when they are infinitesimally down-weighted, as measured by differentiation with respect to the weight assigned to that sample.}
It is computationally lightweight and applicable well beyond linear regression.
Its authors demonstrate that it finds small, high-influence subsets in several large-scale econometric studies.
ZAMinfluence has already seen significant adoption: since its initial release 2021, several studies in economics, finance, and political science have used it to perform robustness checks, ensuring that it \emph{doesn't} find a small subset of samples which can be dropped to change the study outcome \cite{eubank2022enfranchisement,martinez2022much,finger2022adoption,turnbull2022mobilising,falck2022systematic}.

However, ZAMinfluence only provides upper bounds on $\text{Stability}(X,y)$.
Moitra and Rohatgi show in the context of the Boston Housing dataset that these upper bounds frequently are not tight.
We show in this work that this non-tightness extends to multiple cases where ZAMinfluence has been used as a (purported) robustness check.

Kuschnig, Zens, and Crespo Cuaresma \cite{kuschnig2021hidden} experiment with several refinements of ZAMinfluence, mainly involving removing the most influential sample one at a time and then recomputing all influences.
This amounts to a greedy heuristic for approximating Stability.
They show that this heuristic finds better upper bounds on Stability than ZAMinfluence does in several examples.

\paragraph{\textsc{PartitionAndApprox} and \textsc{NetApprox}}
Moitra and Rohatgi \cite{moitra2022provably} propose two algorithms to approximate a fractional variant of $\text{Stability}(X,y)$ (meaning they search for a set of $[0,1]$-valued \emph{weights} rather than a subset $S$, and consider the resulting weighted OLS solution).
They prove strong theoretical guarantees for their algorithms -- in particular, under relatively weak assumptions on $X,y$, they can compute the fractional stability, up to error $\eps n$, in time roughly $(n/\eps)^{d + O(1)}$.

Moitra and Rohatgi implement modified variants of \textsc{PartitionAndApprox} and \textsc{NetApprox} for which their provable guarantees no longer apply, but which still give, for any given $X,y$, valid upper and lower bounds on the fractional stability.
They demonstrate that these lower bounds are nontrivial -- for a majority of two-dimensional regression problems drawn from the Boston Housing dataset \cite{harrison1978hedonic} their upper and lower bounds are within a factor of two.

\textsc{PartitionAndApprox} and \textsc{NetApprox} thus offer a potentially useful robustness audit, but from a practical standpoint there are still significant drawbacks.
First, their running times scales poorly with dimension.
(Indeed, \cite{moitra2022provably} show that this is inherent for any algorithm which provably computes Stability for any $X,y$.)
Consequently, Moitra and Rohatgi do not obtain nontrivial stability bounds on any regression problem of dimension larger than three.
Second, their upper and lower bounds are still far from tight for the majority of regression problems they draw from Boston Housing (e.g., their bounds are not within 1\% of each other on 92\% of the instances). 
By contrast, we solve $94\%$ of these problems with gaps of less than $1\%$.

\paragraph{Additional Prior Work}
Measures of the influence of individual samples on a linear regression have been extensively studied in statistics.
This literature is too broad to fully survey here; see e.g. \cite{chatterjee1986influential} and references therein for discussion of classical literature.
Determining what a fitted model would do in the absence of a subset of data has also been of recent interest in machine learning \cite{ilyas2022datamodels,yang2023many}.


\subsection{Results}
We now briefly summarize the upper and lower bounds we obtain on Stability for real-world and synthetic datasets, and to what extent these improve on prior work.
We report full results in Section~\ref{sec:experiments}, where we run every algorithm that we study on every dataset, except where (a) the algorithm only works on a special case of linear regression which the dataset doesn't fit, or (b) the algorithm requires too much time to return results (see discussion in Section ~\ref{ssec:setup}) , as for \textsc{PartitionAndApprox}/\textsc{NetApprox} on datasets of dimension $3$ or larger.

\paragraph{Microcredit}
Meager \cite{meager2022aggregating} surveys seven randomized control trials involving availability of microcredit loans in developing countries.
Each involves a single regression of one binary treatment variable and a real-valued outcome, typically with thousands or tens of thousands of samples.
They are among the original datasets investigated using ZAMinfluence \cite{broderick2020automatic}.
Both our Gurobi-based approach and a simple greedy algorithm exactly solve Stability for all these studies (that is, they obtain matching lower and upper bounds).
In several cases our upper bounds improve on those obtained via ZAMinfluence,  and we provide the first lower bounds for these datasets.

\paragraph{Incarceration}
Eubank and Fresh \cite{eubank2022enfranchisement} investigate the effect of the end of Jim Crow on incarceration rates of Black people in the American South.
The resulting linear regression is 48-dimensional with 504 samples.
Using ZAMinfluence, Eubank and Fresh report that at least 19\% of their data would need to be removed to change the outcome of their study.
Our Gurobi-based method identifies a subset of $< 6\%$ of the data which has this effect.
None of our algorithms find any nontrivial lower bound on Stability for these data.

\paragraph{GDP and Democracy}
Martinez \cite{martinez2022much} investigates the effect of political freedom on national reports of economic growth.
The resulting linear regression is 211-dimensional, with 3895 samples.
\cite{martinez2022much} reports that ZAMinfluence needs to remove at least $5\%$ of the data to change the study's outcome; the authors run some heuristic tests to try and see if this $5\%$ can be reduced to $1\%$ and report that it likely cannot.
However, using Gurobi we find a subset of $\approx 3\%$ of the sample which can be removed to change the outcome of the study.
None of our algoirthms provide any nontrival lower bound on Stability for these data.

\paragraph{Boston Housing}
Boston Housing \cite{harrison1978hedonic} is a standard benchmark dataset for machine learning.\footnote{Some ethical issues surrounding the Boston Housing data have emerged in recent years \cite{fairlearn_boston_housing_data}. We report no conclusions or predictions made using these data, only Stability of some regression problems drawn from them, in order to compare our algorithms with \cite{moitra2022provably}.}
Moitra and Rohatgi test \textsc{PartitionAndApprox} and \textsc{NetApprox} on numerous regression problems with two-dimensional covariates drawn from Boston Housing.
On nearly all of these regression problems, Gurobi finds upper and lower bounds on Stability with significantly smaller gap (typically less than $1\%$).
On a few examples, Moitra and Rohatgi's algorithm obtains tighter bounds; in practice one could run both algorithms and report the tighter bound.

\paragraph{Minimum Wage}
Card and Krueger \cite{card1993minimum} study the effect of minimum wage on fast-food employment.
The resulting data and analysis have become a textbook example of the difference-in-differences method (indeed, our analysis is based on a CSV prepared by \cite{bauer2020causal}).
We design a simple greedy method which can exactly compute Stability for such difference-in-differences regressions; it scales easily to the 384 observation pairs in Card and Krueger's dataset.

\paragraph{Synthetic}
We expose an Achilles heel of all previously-discussed algorithms for robustness auditing of linear regression: none can provide nontrivial lower bounds on Stability for simple synthetic datasets with very robust linear trends.
Concretely, we generate $1000$ i.i.d. samples $X_1,\ldots,X_{1000} \in \R^4$ from $\mathcal{N}(0,I)$ and we let $Y_i = \beta^\top X_i + \eps_i$, where $\eps_i \sim \mathcal{N}(0,1)$ and $\beta = (1,1,1,1)$.
Using Gurobi we can identify a subset of ?? samples to remove, but we obtain no nontrivial lower bound.
However, our spectral algorithm supplies a lower bound showing that at least $10 \%$ of the data must be removed to change the sign of $\beta_1$.
\section{Our Algorithms}\label{sec:algos}

Now we give formal descriptions of the algorithms we study.
We defer to \cite{moitra2022provably} for descriptions of ZAMinfluence, \textsc{PartitionAndApprox}, and \textsc{NetApprox}.

\subsection{Mathematical Programs for Robustness Auditing}\label{ssec:gurobi}
Moitra and Rohatgi observe that a fractional version of robustness auditing for linear regression with $n$ samples in $d$ dimensions can be cast as a mathematical program with $n$ fractional $[0,1]$ variables, $d$ real-valued variables, bilinear constraints which are linear in the binary and real-valued variables respectively, and a linear objective function.  Effectively, they consider a version of the problem in which weights do not have to be binary but may instead be chosen as fractional values and the last OLS coefficient $\beta_d$ is constrained to be 0, and write it as follows

\begin{equation*}
\begin{aligned}
&&&n\quad-\quad\max_{\boldsymbol{\beta}\in \mathrm{R}^{d-1}, \boldsymbol{w}\in[0,1]^n} \quad |\boldsymbol{w}|_1  \\
&&&\text{s.t. } 
\sum_{i=1}^nw_iX_{i,j'}\left(\sum_{j=1}^{d-1} X_{i,j}\beta_j - Y_i\right) = 0 \quad \forall j'\in[d]
\end{aligned}
\end{equation*}

The $d$ bilinear constraints together enforce that the gradient of the OLS squared-error is $0$ at $\beta$ for the regression instance specified by the weights $w$ -- i.e, that $\beta$ is an OLS regressor for this weighted regression problem.
The constraints implicitly ensure that $\beta_d$, the last coefficient of $\beta$, is $0$, because the residual term $\sum_{j=1}^{d-1} X_{i,j} \beta_j - Y_i$ does not include any $\beta_d$ term.

The previous mathematical program solves a fractionally relaxed version of the stability problem. In this relaxed variant it is guaranteed that an optimal solution would set $\beta_d=0$.  For the integral problem we need to include an explicit variable for $\beta_d$; we can then write
\begin{equation*}
\begin{aligned}
&&&n\quad-\quad\max_{\boldsymbol{\beta}\in \mathrm{R}^{d}, \boldsymbol{w}\in\{0,1\}^n, \boldsymbol{r}\in \mathrm{R}^n} \quad |\boldsymbol{w}|_1  \\
&&&\text{s.t. } 
\sum_{i=1}^nw_iX_{i,j'}\left(\sum_{j=1}^d X_{i,j}\beta_j - Y_i\right) = 0 \quad \forall j'\in[d]\\ 
&&&
\qquad\beta_d\leq 0\\
\end{aligned}
\end{equation*}

Though these problems are nonconvex, with the terms $w_i\beta_j$ appearing in the constraints, they can be solved using exact solver methods, including those supported from Gurobi 9.0 onwards \cite{achterberg2020s}. In particular, these methods apply a globally optimal spatial branch-and-bound method which recursively partitions the feasible region into subdomains and invokes McCormick inequalities to obtain lower and upper bounds within each subdomain. We refer the reader to \cite{belotti2013mixed} for an overview of the general theory underlying these methods.

\paragraph{Implementation details.} 
Given unconstrained runtime, Gurobi is guaranteed to solve both the fractional and the integral quadratically constrained optimization problems. In practice, we find on all of our instances that Gurobi identifies good solutions much quicker for the fractional problem (for low-dimensional problems this entails provably small error, for high-dimensional problems the best heuristic solutions we can identify).  In all of our instances the fractional solution can be easily rounded to an integral one by just rounding every weight that is strictly smaller than 1 (with some numerical tolerance) to 0. We then run OLS on the subset of samples given by the rounded weights to confirm that the result has a negative final coefficient $\beta_d$. Alternatively, one can warm-start Gurobi on the integer-constraint instance using the rounded weights; however, we have found no instances where this gives improved solutions.

Directly running Gurobi on the integer-constrained instance, without a warm start obtained by rounding a fractional solution, often shows significantly worse performance (e.g., on Eubank and Fresh's data we can identify an upper bound of 28, by rounding the fractional solution, within seconds, whereas the integer-constrained optimization takes more than 30 minutes to identify an upper bound of 187).

On one of the Microcredit instances \cite{angelucci2015microcredit}, we identified an idiosyncratic behavior of the Gurobi solver: it returns an incorrect (claimed optimal) upper bound of $|\boldsymbol{w}|_1=0$ when solving the fractional problem. However, with the added constraint $|\boldsymbol{w}|_1\geq1$, Gurobi solves the fractional instance optimally. That constraint affects the performance on other instances, so we only include it when not including it leads to an incorrect upper bound of 0.


    

\subsection{Nearly-Linear Time Algorithm for Single Binary Treatment Variable}\label{ssec:binary_treatment}
For the simplest regression problems, with a single binary treatment variable and a real-valued outcome, we show that Stability is computable in time $O(n \log n)$.
The algorithm below assumes that the OLS regression for the input $(X,Y)$ has positive slope; otherwise replace each $Y_i$ by $-Y_i$.

The simple insight behind the algorithm is that if each $X_i \in \{0,1\}$ and we commit to removing $k \leq n$ samples, then the best subset of samples to remove to minimize the slope of the regression line consists of samples with $X_i = 0$ and minimum $Y_i$s, and samples with $X_i = 1$ and maximum $Y_i$s.

\begin{algorithm}
\caption{Exact Algorithm For Auditing Binary 2D Regression}
\label{alg:binary}
\begin{algorithmic}[1]
\Procedure{RobustnessAuditBinary}{$X,Y$}
    \State $n \gets \text{length}(X)$
    \State Let $Y^i = \{Y_j | X_j = i\}$ for $i \in \{0, 1\}$, sort $Y^0$ in decreasing, $Y^1$ in increasing order
    \State Set  $S^0_\ell$ as cumulative sums of the first $\ell$ terms of  $Y^0$ for $\ell\in 1,\ldots,|Y^0|$
    \State Set  $S^1_\ell$ as cumulative sums of the first $\ell$ terms of  $Y^1$ for $\ell\in 1,\ldots,|Y^1|$

    \State $lower, upper\gets 0, n$ \Comment{$lower$ is too small and $upper$ is sufficient to flip sign}
    \While{\textsc{True}}
        \State $\textsc{Flag}\gets\textsc{False}$; \quad $k\gets \lfloor (lower+upper)/2\rfloor$
        \For{$\ell \gets \max\{0,n-k-|Y^1|\}$ to $\min\{|Y^0|,n-k\}$} \Comment{Iterate over number of 0s to drop}
        \State \textbf{if} $-(n-k-\ell)*S^0_\ell + \ell* S^1_{n-k-\ell}\leq 0$ \textbf{do} $\textsc{Flag}\gets\textsc{True}$
        \EndFor
        \State \textbf{if} $lower=k=upper-1$ and \textsc{Flag}: \quad \Return $k$
        \State \textbf{if} $lower=k=upper-1$ and \textsc{not Flag}: \quad \Return $k+1$
        \State \textbf{if} $\textsc{Flag}$: $upper\gets k$;\quad \textbf{else}: $lower\gets k$
    \EndWhile
\EndProcedure
\end{algorithmic}
\end{algorithm}

We capture correctness of Algorithm~\ref{alg:binary} in the following theorem.

\begin{theorem}
  Given $X_1,\ldots,X_n \in \{0,1\}$ and $Y_1,\ldots,Y_n \in \R$, Algorithm~\ref{alg:binary} outputs $\text{Stability}(X,Y)$ in time $O(n \log n)$.
\end{theorem}
\begin{proof}
  To prove correctness of Algorithm~\ref{alg:binary}, we need to show two things.
  First, to correctly apply binary search, we need to show a monotonicity property: if there is a subset of $k$ samples which we can remove to change the OLS slope then for every $k' > k$ there is also such a subset of $k'$ samples.
  Second, we need to show correctness of the greedy step: if there is a subset of $k$ samples which change the sign of the slope of the OLS regression line when removed, then there exists such a subset which, for some $\ell \leq k$, removes the $\ell$ samples such that $X_i = 0$ with least $Y_i$ and the $k-\ell$ samples with $X_i = 1$ and greatest $Y_i$.

  For both these goals, let $S_1 = \{i \in [n] \, : \, X_i = 1\}$ and consider some subset $U \subseteq [n]$ with $|U| = n-k$.
  We derive an explicit formula for the slope of the OLS line on the dataset $\{(X_i,Y_i)\}_{i \in U}$.
  \begin{align*}
  \beta & = \Paren{\sum_{i \in U} (1,X_i) (1,X_i)^\top}^{-1} \sum_{i \in U} (1,X_i) \cdot Y_i \\
  & = \Paren{
  \begin{matrix} |U| & |U \cap S_1| \\
  |U \cap S_1| & |U \cap S_1|
  \end{matrix} }^{-1} \cdot \Paren{\begin{matrix} \sum_{i \in U} Y_i \\ \sum_{i \in U} X_i Y_i \end{matrix}} \\
  & = \frac 1 { |U| |U \cap S_1| - |U \cap S_1|^2} \cdot \Paren{\begin{matrix} |U \cap S_1|  & -|U \cap S_1| \\ -|U \cap S_1| & |U| \end{matrix} } \cdot \Paren{\begin{matrix} \sum_{i \in U} Y_i \\ \sum_{i \in U} X_i Y_i \end{matrix}} \, .
  \end{align*}
  The sign of the second coordinate of $\beta$ (which gives the slope) is the same as the sign of
  \[
  - |U \cap S_1| \cdot \sum_{i \in U} Y_i + |U| \cdot \sum_{i \in U \cap S_1} Y_i = - |U \cap S_1| \cdot \sum_{i \in U\cap S_0} Y_i + |U\cap S_0|\sum_{i \in U\cap S_1} Y_i
  \]
  Among all $U$ with a given $|U|$ and $|U \cap S_1|$, this expression is clearly minimized by minimizing $\sum_{i \in U \cap S_1} Y_i$ and maximizing $\sum_{i \in U \cap S_0} Y_i$.
  This establishes correctness of the greedy step.

  Now suppose that the OLS slope on $U$ is non-positive; we need to show that the same holds for some $U'$ with $|U'| = |U| - 1$, to establish correctness of binary search.
  By the above, we can assume
  \[
  -|U \cap S_1| \cdot \sum_{i \in U} Y_i + |U| \cdot \sum_{i \in U \cap S_1} Y_i \leq 0
  \]
  which rearranges to
  \[
  \frac{\sum_{i \in U \cap S_1} Y_i}{|U \cap S_1|} \leq \frac{\sum_{i \in U \cap S_0} Y_i} {|U \cap S_0|} \, .
  \]
  This is clearly preserved by removing from $U$ either one of $i^* = \arg \max_{i \in U \cap S_1} Y_i$ and $j^* = \arg \min_{j \in U \cap S_0} Y_j$.
\end{proof}

\subsection{Nearly-Linear Time Algorithm for Difference-in-Differences}\label{ssec:diffs}
We study the following difference-in-differences regression setting.
$N$ individuals in two groups, treatment and non-treatment, each report two responses, $Y_{i,\text{before}},Y_{i,\text{after}} \in \R$.
Here ``before'' and ``after'' correspond, respectively, to before and after the time at which the treatment group is treated.
The difference-in-differences linear model is then
\[
Y = \beta_0 + \beta_{1} \cdot \text{time} + \beta_{2} \cdot \text{treatment} + \beta_{3} \cdot \text{time} \times \text{treatment} \,
\]
where ``time'', ``treatment'' assume values in $\{0,1\}$, and the coefficient of interest is $\beta_3$.
We study Stability with respect to the removal of individuals from the dataset -- note that removing an individual corresponds to removing two data points, ``before'' and ``after''.

The following lemma motivates our algorithm; it is based on a standard closed-form expression of the difference-in-difference regressor and is proved for completeness in Section~\ref{sec:proof-lem-diff-in-diffs}. 

\begin{lemma}
\label{lem:diff-in-diffs-reformulation}
    Let $Y_{1,\text{before}},Y_{1,\text{after}},\ldots,Y_{N,\text{before}},Y_{N,\text{after}} \in \R$ be a difference-in-differences dataset with $N$ individuals of which a subset $T \subseteq [N]$ are treated.
    For a subset $U \subseteq [N]$ of individuals, the coefficient $\beta_3^U$ of difference-in-differences on the dataset $U$ has the same sign as
    \[
    \E_{i \sim U \cap T} (Y_{i,\text{after}} - Y_{i,\text{before}} ) -  
    \E_{i \sim U} (Y_{i,\text{after}} - Y_{i,\text{before}}) \, .
    \]
\end{lemma}

Now we can state our algorithm.
The algorithm below assumes that $\beta_3$ on the whole dataset is nonnegative; otherwise simply negate all the $Y$s.

\begin{algorithm}
\caption{Exact Algorithm For Auditing Difference-in-Differences}
\label{alg:diff-in-diffs}
\begin{algorithmic}[1]
\Procedure{RobustnessAuditDiffInDiffs}{$(Y_{1,\text{before}}, Y_{1,\text{after}}),\ldots, (Y_{N,\text{before}}, Y_{N,\text{after}})$, $T \subseteq [N]$}
    \State $\delta_i \gets Y_{i,\text{after}} - Y_{i,\text{before}}\quad\forall i \in [N]$;\qquad $\Delta_{T} \gets \{\delta_i \, : \, i \in T\}$; \qquad $\Delta_{\overline{T}} = \{ \delta_i \, : \, i \in [N] \setminus T\}$.
    \State Sort $\Delta_T$ in increasing and $\Delta_{\overline{T}}$ in decreasing order
    \State Store partial sums $S^T_\ell$ defined as the sum of the first $\ell$ terms of $\Delta_T$ for $\ell\leq |T|$ 
    \State Store partial sum $S^{\overline{T}}_\ell$ defined as the sum of the first $\ell$ terms of $\Delta_{\overline{T}}$ for $\ell\leq N-|T|$

    \State $lower, upper\gets 0, N$ \Comment{$lower$ is too small and $upper$ is sufficient to flip sign}
    \While{\textsc{True}}
        \State $\textsc{Flag}\gets\textsc{False}$; \quad $k\gets \lfloor (lower+upper)/2\rfloor$
        \For{$\ell \gets \max\{0,k-|T|-N+1\}$ to $\min\{|T|-1,k\}$} \Comment{Iterate over \# of treated to remove}
        \State \textbf{if} $\frac{S_{|\Delta_T|-\ell}^{T}}{|T|-\ell}
        -\frac{S_{|\Delta_T|-\ell}^{T}+S_{|\Delta_{\overline{T}}|-(k-\ell)}}{N-k}
        \leq 0$ \textbf{do} $\textsc{Flag}\gets\textsc{True}$
        \EndFor
        \State \textbf{if} $lower=k=upper-1$ and \textsc{Flag}: \quad \Return $k$
        \State \textbf{if} $lower=k=upper-1$ and \textsc{not Flag}: \quad \Return $k+1$
        \State \textbf{if} $\textsc{Flag}$: $upper\gets k$;\quad \textbf{else}: $lower\gets k$
    \EndWhile
\EndProcedure
\end{algorithmic}
\end{algorithm}

We show:
\begin{theorem}
    Algorithm~\ref{alg:diff-in-diffs} runs in $O(N\log N)$-time algorithm, taking $N$ individuals $\{(Y_{i,\text{before}}, Y_{i,\text{after}})\}_{i \in N}$, divided into treated and untreated subgroups, and returns the size of a minimum-size set $U \subseteq [N]$ such that the sign of $\beta_3$ for the dataset $U$ differs from the sign of $\beta_3$ on the dataset $[N]$.
\end{theorem}

\begin{proof}
    Running time is straightforward: lines 4,5 each take $O(N \log N)$ time to sort and $O(N)$ time for the partial sums.
    Then binary search requires $O(\log N)$ rounds, each requiring $O(N)$ times through the for loop on line 8; each execution of the loop is constant time using the stored partial sums.

    For correctness, we need to argue two things.
    First, we must show that if there is a subset of $k$ individuals which can be removed to make $\beta_3$ negative, then there is also such a subset of $k+1$ individuals -- this gives correctness of binary search.
    Second, we must show that if there is such a subset of $k$ individuals, then for some $\ell \leq k$ it can be found by removing the $\ell$ treated individuals with greatest $\delta_i$ and the $k-\ell$ non-treated individuals with least $\delta_i$.

    We start with monotonicity.
    Suppose that some subset of individuals $U \subseteq [N]$ has non-positive $\beta_3$.
    Then by Lemma~\ref{lem:diff-in-diffs-reformulation}, $\E_{i \sim U \cap T} \delta_i \leq \E_{i \sim U} \delta_i$.
    If some non-treated individual $j \in U$ has $\delta_j \leq \E_{i \sim U} \delta_i$, then we can remove $j$ from $U$ and maintain the inequality $\E_{i \sim (U \setminus \{j\}) \cap T} \delta_i \leq \E_{i \sim (U \setminus \{j\})} \delta_i$.

    Otherwise, every non-treated individual $j \in U$ has $\delta_j > \E_{i \sim U} \delta_i$.
    Now we remove any non-treated individual $j$ and obtain
    \[
    \E_{i \sim U \setminus \{j\}} \delta_i
    = \frac{|U \cap T|}{|U|-1} \cdot \E_{i \sim U \cap T} \delta_i + \frac{|U \cap \overline{T}|-1}{|U|-1} \cdot \E_{i \sim U \cap \overline{T} \setminus \{j \}} \delta_i
    \geq \frac{|U \cap T|}{|U|-1} \cdot \E_{i \sim U \cap T} \delta_i + \frac{|U \cap \overline{T}|-1}{|U|-1} \cdot \E_{i \sim U} \delta_i 
    \]
    where for the last inequality we used that every nontreated $j \in U$ has $\delta_j > \E_{i \sim U} \delta_i$.
    By hypothesis, $\E_{i \sim U} \delta_i \geq \E_{i \sim U \cap T} \delta_i$, so overall we got $\E_{i \sim U \setminus \{j\}} \delta_i \geq \E_{i \sim U \cap T} \delta_i$ as desired.

    Correctness of the greedy step is clear from Lemma~\ref{lem:diff-in-diffs-reformulation}.
\end{proof}

\subsubsection{Proof of Lemma~\ref{lem:diff-in-diffs-reformulation}}
\label{sec:proof-lem-diff-in-diffs}
We turn to the proof of Lemma~\ref{lem:diff-in-diffs-reformulation}, starting with some setup.

We reformulate difference-in-differences as an OLS regression problems with the usual $(X_i,Y_i)$ pairs.
Each individual contributes two vectors $X_{i,\text{before}}, X_{i,\text{after}} \in \{0,1\}^4$, where the first coordinate corresponds to the intercept of the regression line and the remaining three coordinates are time, treatment, and time $\times$ treatment, respectively. That is,
\begin{align*}
& X_{i,\text{before}}(0) = 1\\
& X_{i,\text{before}}(1) = 0\\
& X_{i,\text{before}}(2) = 0 \text{ if $i$ is not treated and otherwise } 1\\
& X_{i,\text{before}}(3) = 0\\
& X_{i,\text{after}}(0) = 1\\
& X_{i,\text{after}}(1) = 1\\
& X_{i,\text{after}}(2) = 0 \text{ if $i$ is not treated and otherwise } 1\\
& X_{i,\text{after}}(3) = 0 \text{ if $i$ is not treated and otherwise } 1 \, .
\end{align*}

We need one definition:
\begin{definition}
  Let $n \in \N$ and $s \in \N$ with $s \leq n$.
  The $(n,s)$-diff-in-diff covariance matrix is:
  \[
  \Sigma_{n,s} = \left ( \begin{matrix}
      n          & n/2 & s          & s/2 \\
      n/2 & n/2 & s/2 & s/2 \\
      s          & s/2 & s          & s/2 \\
      s/2 & s/2 & s/2 & s/2
  \end{matrix} \right ) \, .
  \]
  Note that for a diff-in-diff dataset with $n/2$ individuals and hence $n$ samples $X_i$, where $s/2$ of those individuals are in the treatment group, $\Sigma_{n,s} = \sum_{i \leq n} X_i X_i^\top$.
\end{definition}

The following is easy to check in e.g. Mathematica:

\begin{fact} 
\[
\det \Sigma_{n,s} \cdot \Sigma_{n,s}^{-1} =
\frac 1 8 \cdot \left(
\begin{array}{cccc}
 s^2 (n-s) & s^2 (s-n) & s^2 (s-n) & s^2 (n-s) \\
 s^2 (s-n) & 2 s^2 (n-s) & s^2 (n-s) & 2 s^2 (s-n) \\
 s^2 (s-n) & s^2 (n-s) & n s (n-s) & n s (s-n) \\
 s^2 (n-s) & 2 s^2 (s-n) & n s (s-n) & 2 n s (n-s) \\
\end{array}
\right)
\]
\end{fact}

\begin{proof}[Proof of Lemma~\ref{lem:diff-in-diffs-reformulation}]
Let $U \subseteq [N]$ be a subset of individuals in a diff-in-diffs dataset, where $T \subseteq [N]$ are the treated individuals and where $|U| = m/2$ and $U$ contains $s/2$ treatment individuals.
Then
\[
  \beta^U = \Sigma_{m,s}^{-1} \cdot \sum_{i \in U} X_{i,\text{before}} Y_{i,\text{before}} + X_{i,\text{after}} Y_{i,\text{after}}
\]

Since $\Sigma_{m,s} \succeq 0$ and hence $\det \Sigma_{m,s} \geq 0$, the sign of $\beta^U_3$ is the same as the sign of
\[
(s^2 (m-s), 2s^2 (s-m), ms(s-m), 2ms(m-s))^\top \sum_{i \in U} X_{i,\text{before}} Y_{i,\text{before}} + X_{i,\text{after}} Y_{i,\text{after}} \, ,
\]
which, since $s \geq 0$ and $m-s \geq 0$, has the same sign as
\[
(s, -2s, -m, 2m)^\top \sum_{i \in U} X_{i,\text{before}} Y_{i,\text{before}} + X_{i,\text{after}} Y_{i,\text{after}} \, .
\]
Dividing by $sm$, applying the definition of $X_{i,\text{before}}$ and $X_{i,\text{after}}$, and simplifying, this has the same sign as
\[
\E_{i \sim U} (Y_{i,\text{before}} + Y_{i,\text{after}} ) - 2 \E_{i \sim U} Y_{i,\text{after}} - \E_{i \sim U \cap T} (Y_{i,\text{before}} + Y_{i,\text{after}} ) + 2 \E_{i \sim U \cap T} Y_{i,\text{after}} \, ,
\]
which rearranges to the conclusion of the lemma.
\end{proof}

\subsection{Spectral Algorithm}\label{ssec:spectral}
In this section we describe and analyze our spectral robustness auditor.

\begin{algorithm}
\caption{Spectral Robustness Auditing}
\label{alg:spectral}
\begin{algorithmic}[1]
\Procedure{RobustnessAuditSpectral}{$X,Y$}
    \State $n \gets \text{len}(X)$
    \State $\beta \gets \text{OLS}(X,Y)$
    \State $M_1 \gets $ a $d \times n$ matrix where $i$-th column is $X_i \cdot (\iprod{X_i,\beta} - y_i)$
    \State $\Sigma \gets \tfrac 1 n \sum_{i \leq n} X_i X_i^\top$
    \State $C_1 \gets \|\Sigma^{-1/2} M_1\| / \sqrt{n}$ (maximum singular value of $\Sigma^{1/2} M_1 / \sqrt{n}$)
    \State $M_2 \gets $ a $d^2 \times n$ matrix where the $i$-th column is $\Sigma^{-1/2}X_i \otimes \Sigma^{-1/2} X_i$.
    \State $\Phi \gets $ a $d^2$-length vector where $\Phi_{i,i} = 1$ and $\Phi_{i,j} = 0$ if $i \neq j$.
    \State $W \gets (\tfrac 3 {2+d})^{1/2} \cdot \tfrac{\Phi \Phi^\top}{d} + (\tfrac 3 2)^{1/2} \cdot (I - \tfrac{\Phi \Phi^\top}{d})$.
    \State $C_2 \gets \|W M_2  / \sqrt{n}\|$ (maximum singular value of $W M_2 / \sqrt{n}$)
    \State $\varepsilon \gets \frac {\beta_i^2}{C_1 \cdot \sqrt{\Sigma^{-1}_{i,i}} + C_2 |\beta_i|}$
    
    \Return $\varepsilon$
\EndProcedure
\end{algorithmic}
\end{algorithm}

The key lemma is the following one, which is explicit to varying degrees in prior works such as \cite{klivans2018efficient,bakshi2021robust}.
We provide a short proof for completeness.

\begin{lemma}[Implicit in \cite{bakshi2021robust}]
\label{lem:spectral-correct-main}
  Let $X_1,\ldots,X_n \in \R^d, y_1,\ldots,y_n \in \R$ and let $\beta$ be the solution to OLS on $\{(X_i,y_i)\}_{i \in [n]}$.
  Let $\Sigma = \tfrac 1 n \sum_{i \leq n} X_i X_i^\top$.
  Let $C_1,C_2 \geq 0$ satisfy the following inequalities for every $v \in \R^d$:
  \begin{align*}
      & \frac 1 n \sum_{i \leq n} \iprod{X_i,v}^2 (\iprod{X_i,\beta} - y_i)^2 \leq C_1 \cdot \iprod{v, \Sigma v} \\
      & \frac 1 n \sum_{i \leq n} \iprod{X_i,v}^4 \leq C_2 \cdot \iprod{v, \Sigma v}^2 \, .
  \end{align*}
  Let $S \subseteq [n]$ and let $\beta_S$ be the solution to OLS on $\{(X_i,y_i)\}_{i \in S}$.
  Then
  \[
    \frac{n- |S|}{n} \geq \frac{(\beta_S - \beta)_1^2}{ \left( \sqrt{C_1} \|\Sigma^{-1/2} e_1\| + |(\beta_S - \beta)_1|\sqrt{C_2} \right)^2} \, ,
  \]
  where $(\beta_S - \beta)_1$ is the first coordinate of the vector $(\beta_S - \beta)$.
\end{lemma}

This allows us to prove:

\begin{theorem}
    Algorithm~\ref{alg:spectral} returns a valid lower bound on $\text{Stability}(X,Y)$ using $O(1)$ top singular values of matrices of dimension at most $n \times d^2$, a single $d \times d$ matrix inverse, and additional running time $O(nd^2)$.
\end{theorem}

\begin{proof}
In light of Lemma~\ref{lem:spectral-correct-main}, to prove correctness of Algorithm~\ref{alg:spectral} we just need to show that $C_1$ and $C_2$ as computed in Algorithm~\ref{alg:spectral} satisfy the hypotheses of Lemma~\ref{lem:spectral-correct-main}.
For $C_1$ this is clear by construction.

For $C_2$, we first observe that by replacing $v$ with $\Sigma^{-1/2} v$ we can just as well prove that for all $v \in \R^d$,
\[
\frac 1 n \sum_{i \leq n} \iprod{\Sigma^{-1/2} X_i, v}^4 \leq C_2 \|v\|^4 = C_2 \cdot (v \otimes v)^\top (\tfrac 2 3 I + \tfrac 1 {3} \Phi \Phi^\top) (v \otimes v) \, .
\]
Simple linear algebra shows that the matrix $W$ in Algorithm~\ref{alg:spectral} is exactly $(\tfrac 2 3 I + \tfrac 1 {3} \Phi \Phi^\top)^{-1/2}$.
So replacing $v \otimes v$ with $W^{-1/2} (v \otimes v)$ shows that $\|W M_2 / \sqrt{n}\|$ is a valid choice for $C_2$.
This proves correctness of Algorithm~\ref{alg:spectral}.
The running time is clear from inspection.
\end{proof}

\subsubsection{Proof of Lemma~\ref{lem:spectral-correct-main}}
To prove the lemma we need the following claim, which says $\tfrac 1 n \sum_{i \in S} X_i X_i^\top$ isn't too different from $\Sigma$.

\begin{claim}
\label{clm:covariance-lb}
  Let $X_1,\ldots,X_n,y_1,\ldots,y_n, \Sigma, S$, and $C_2$ be as in Lemma~\ref{lem:spectral-correct-main}.
  Let $\Sigma_S = \tfrac 1 n \sum_{i \in S} X_i X_i^\top$.
  Then 
  \[
  \Sigma_S \succeq \Paren{1 - \sqrt{C_2 \cdot \tfrac{n-|S|}{n}}} \cdot \Sigma \, .
  \]
\end{claim}
\begin{proof}[Proof of Claim~\ref{clm:covariance-lb}]
  Let $v \in \R^d$.
  We have
  \[
  \frac 1 n \sum_{i \in \overline{S}} \iprod{X_i,v}^2 = \frac 1 n \sum_{i \leq n} 1(i \notin S) \cdot \iprod{X_i,v}^2 \leq \sqrt{\frac 1 n \sum_{i \leq n} 1(i \notin S)} \cdot \sqrt{\frac 1 n \sum_{i \leq n} \iprod{X_i,v}^4}
  \leq \sqrt{\frac{n - |S|}{n}} \cdot \sqrt{C_2} \cdot \iprod{v, \Sigma v} \, ,
  \]
  where the inequality is Cauchy-Schwarz.
  So,
  \[
  \iprod{v, \Sigma_S v} = \frac 1 n \sum_{i \in S} \iprod{X_i,v}^2 = \iprod{v, \Sigma v} - \frac 1 n \sum_{i \in \overline{S}} \iprod{X_i,v}^2 \geq \Paren{1 - \sqrt{C_2 \cdot \frac{n - |S|}{n}}} \cdot \iprod{v, \Sigma v} \, ,
  \]
  which is what we wanted to show.
\end{proof}

\begin{proof}[Proof of Lemma~\ref{lem:spectral-correct-main}]
  Let $\Sigma_S = \tfrac 1 n \sum_{i \in S} X_i X_i^\top$.
  We start by bounding $\|\Sigma_S^{1/2} (\beta_S - \beta)\|^2$:
  \begin{align*}
    \|\Sigma_S^{1/2} (\beta_S - \beta)\|^2 & = \iprod{\beta_S - \beta, \Paren{\frac 1 n \sum_{i \in S} X_i X_i^\top } (\beta_S - \beta) } \\
    & = \iprod{\beta_S - \beta, \frac 1 n \sum_{i \in S} X_i (\iprod{X_i, \beta_S} - y_i) - \frac 1 n \sum_{i \in S} X_i (\iprod{X_i, \beta} - y_i)} \text{ by adding and subtracting $X_i y_i$} \\
    & = \iprod{\beta_S - \beta, - \frac 1 n \sum_{i \in S} X_i (\iprod{X_i, \beta} - y_i)} \text{ since $\beta_S$ minimizes $\sum_{i \in S} (\iprod{X_i, \beta_S} - y_i)^2$.} \\
    & = \iprod{\beta_S - \beta, - \frac 1 n \sum_{i \leq n} X_i (\iprod{X_i, \beta} - y_i) + \frac 1 n \sum_{i \in \overline{S}} X_i (\iprod{X_i, \beta} - y_i) } \\
    & = \iprod{\beta_S - \beta, \frac 1 n \sum_{i \in \overline{S}} X_i (\iprod{X_i, \beta} - y_i)} \text{ since $\beta$ minimizes $\sum_{i \leq n} (\iprod{X_i, \beta} - y_i)^2$.}
  \end{align*}
  The last expression above we can bound via Cauchy-Schwarz.
  It is equal to
  \begin{align*}
  \frac 1 n \sum_{i \leq n} 1(i \notin S) \cdot \iprod{\beta_S - \beta, X_i} \cdot (\iprod{X_i,\beta} - y_i)  & \leq \sqrt{\frac{n-|S|}{n}} \cdot \sqrt{\frac 1 n \sum_{i \leq n} \iprod{\beta_S - \beta,X_i}^2 (\iprod{X_i, \beta} - y_i)^2 } \\
  & \leq \sqrt{\frac{n - |S|}{n}} \cdot \sqrt{C_1} \cdot \sqrt{ \iprod{\beta_S - \beta, \Sigma (\beta_S - \beta)}}\\
  & = \sqrt{\frac{n - |S|}{n}} \cdot \sqrt{C_1} \cdot \|\Sigma^{1/2} (\beta_S - \beta)\| \, ,
  \end{align*}
  where the second inequality uses that $\frac 1 n \sum_{i \leq n} \iprod{X_i,v}^2 (\iprod{X_i,\beta} - y_i)^2 \leq C_1 \cdot \iprod{v, \Sigma v}$ for every $v\in\R^d$.
  Overall, we have obtained
  \begin{align}
      \label{eq:spectral-1}
      \| \Sigma_S^{1/2} (\beta_S - \beta)\|^2 \leq \sqrt{\frac{n-|S|}{n}} \cdot \sqrt{C_1} \cdot \|\Sigma^{1/2} (\beta_S - \beta)\| \, .
  \end{align}
  On the other hand, using Claim~\ref{clm:covariance-lb}, we have
  \begin{align}
  \label{eq:spectral-2}
  \|\Sigma_S^{1/2} (\beta_S - \beta)\|^2 \geq \Paren{ 1 - \sqrt{C_2 \cdot \frac{n - |S|}{n}}} \cdot \|\Sigma^{1/2} (\beta_S - \beta)\|^2 \, 
  \end{align}
  So, putting together \eqref{eq:spectral-1} and \eqref{eq:spectral-2} and dividing both sides by $\left(1 - \sqrt{C_2 \cdot \tfrac{n - |S|}{n}}\right) \cdot \|\Sigma^{1/2} (\beta_S - \beta)\|$,
  \[
  \|\Sigma^{1/2} (\beta_S - \beta)\| \leq \sqrt{\frac{n - |S|}{n}} \cdot \sqrt{C_1} \cdot \frac 1 {1 - \sqrt{C_2 \cdot \frac{n - |S|}{n}}} \, .
  \]
  Finally, the first coordinate of $(\beta_S - \beta)$ is
  \[
  |(\beta_S - \beta)_1| = \left | \iprod{ \Sigma^{-1/2} e_1, \Sigma^{1/2} (\beta_S - \beta)} \right | \leq \|\Sigma^{-1/2} e_1\| \cdot \|\Sigma^{1/2} (\beta_S - \beta)\|
  \]
  where $e_1$ is the first standard basis vector. With $\e = \frac{n - |S|}{n}$, we have obtained
  \[
  |(\beta_S - \beta)_1 | \leq \|\Sigma^{-1/2} e_1\| \cdot \sqrt{\e} \cdot \sqrt{C_1} \cdot \frac{1}{1 - \sqrt{C_2 \e}} \, .
  \]
  Solving for $\e$, we get
  \[
  \e \geq \frac{(\beta_S - \beta)_1^2}{\left(\sqrt{C_1} \|\Sigma^{-1/2} e_1\| + |(\beta_S - \beta)_1|\sqrt{C_2}\right)^2} \, . \qedhere
  \]
\end{proof}

\section{Experiments}\label{sec:experiments}

In this section we discuss the setup of our algorithms, and their results, on a range of case studies. The lower and upper bounds on stability, obtained through our algorithms as well as other existing algorithms, are summarized in Table \ref{tab:table1}. Table \ref{tab:table2} summarizes the run time of the algorithms on each instance. 

\subsection{Experimental setup}\label{ssec:setup}
All our experiments were run on recent commodity laptop hardware (Macbook Air 2022, M2 processor, 24GB RAM), using standard Python libraries (\texttt{numpy}, \texttt{gurobipy}).
We used the following implementations of the algorithms:
\begin{itemize}
    \item[ZAMinfluence \cite{broderick2020automatic}:] Our own implementation, since \cite{broderick2020automatic}'s is implemented in R. See \texttt{auditor\_tools}.
    \item[Greedy Heuristic \cite{kuschnig2021hidden}:] Our own implementation: we find some numerical instability in \cite{moitra2022provably}'s implementation of the greedy heuristic for ill-conditioned or rank-deficient regressions, arising from the use of \texttt{numpy.linalg.inv} for matrix inversions of ill-conditioned matrices, rather than using pseudoinverses via \texttt{numpy.linalg.pinv}.
    \item[\textsc{PartitionAndApprox, NetApprox} \cite{moitra2022provably}:] Implementation provided by \cite{moitra2022provably}.
    We run these algorithms only when we expect both of them to terminate within $5$ minutes. Because their running time scales exponentially with dimension, on our hardware this typically constrains them to $3$ dimensions or fewer.
    \item[Gurobi:] Our implementation, calling the Gurobi mathematical programming solver via \texttt{gurobipy}. We typically cut off the solver after $< 10$ seconds of solving time. Note that total running time is typically 0-5 minutes; in most cases this is dominated by the time to set up the mathematical program in Gurobi.
\end{itemize}

For the Exact 2D Binary, the Exact Difference-in-Differences, and the Spectral algorithms we rely on our own implementation.

\subsection{Results}

\begin{table}[h!]
  \begin{center}
    \caption{Table of lower and upper bounds achieved by each algorithm. Cells left empty correspond to no nontrivial bound having been identified by the algorithm, whereas a dash (--) corresponds to the algorithm not being applicable to a given setting using a reasonable amount of time (e.g., MR22 exceeds our running time limits in high dimensions such as the study by \cite{eubank2022enfranchisement} and our exact algorithms only apply to particular instances). In the right-most column, $n$ denotes the number of samples and $d$ denotes the dimension of the samples, including intercept. (I.e. regression to find a slope and intercept with a single treatment variable has $d=2$.)}
    \label{tab:table1}
    \begin{tabular}{|l|c|c|r r||r r|c|c|}
    \hline
    &
    \multicolumn{1}{c|}{\cite{broderick2020automatic}}
    & \multicolumn{1}{c|}{
    \cite{kuschnig2021hidden}}
    & \multicolumn{2}{c||}{\cite{moitra2022provably}}
    & \multicolumn{2}{c|}{Gurobi}
    & \multicolumn{1}{c|}{Exact}
    & {Spectral}\\
    \hline
    Study/Instance (n,d) & UB  & UB &
    LB & UB &
    LB & UB &
     LB$=$UB&
    LB \\
    \hline
    Bosnia (1195,2) \cite{augsburg2015impacts} & 14 & 13 &
     & 14.8 &
    13 & 13 &
    13 & 3
     \\
    Ethiopia (3113,2) \cite{tarozzi2015impacts} & 1  & 1 &
     & 2 &
    1 & 1 &
    1 &
     \\
    India (6863,2) \cite{banerjee2015miracle} & 6  & 6 &
    4.6 & 5.7 &
    6 & 6 &
    6 & 2
     \\
    Mexico (16560,2) \cite{angelucci2015microcredit} & 1  & 1 &
     & 356 &
    1 & 1 &
    1 &
     \\
    Mongolia (961,2) \cite{attanasio2015impacts} & 16  & 15 &
    13.4 & 19.8 &
    15 & 15 &
    15 & 2
     \\
    Morocco (5498,2) \cite{crepon2015estimating} & 11  & 11 &
    10.4 & 10.5 &
    11 & 11 &
    11 & 2
     \\
    Philippines (1113,2) \cite{karlan2011microcredit} & 9  & 9 &
    7.8 & 9.9 &
    9 & 9 &
    9 &
     \\ \hline
    Min wage (384$\times$2,4) \cite{card1993minimum} & --
     & -- & --
     & -- &
      6 & 10 & 10
     & --
      \\ \hline
    Incarceration (504,48) \cite{eubank2022enfranchisement} & 33
     & 29 & --
     & -- & 
     & 28 & --
     & 
      \\ \hline
    GDP (3895,211)  \cite{martinez2022much}& 136
     & 110 & --
     & -- &  
     & 110  & --
     & 
      \\ \hline
      Synthetic 2D (100,2) & 
     & 63 & 60.2
     & 63 & 63
     & 63  & --
     & 19.5
      \\
      Synthetic 4D (1000,4) & 922 
     & 409 & --
     & 452 & 
     & 409  & --
     & 102
      \\ \hline
    \end{tabular}
  \end{center}
\end{table}

The first seven rows of Table \ref{tab:table1} consider the microcredit studies; as these are based on $X_i\in\{0,1\}$, our Algorithm \ref{alg:binary} obtains optimal results for these. Gurobi also finds optimal results, both for the fractional weights studied by \cite{moitra2022provably} and the integral weights we focus on. In two of the seven settings our results find a smaller set to flip the sign than that identified by ZAMinfluence \cite{broderick2020automatic}. In the other five, our results certify that their upper bound is indeed optimal. The reuslts of \cite{moitra2022provably} on this data do not provide comparably strong bounds, despite taking significantly longer to run, as displayed in the first row of Table \ref{tab:table2}. They only find nontrivial lower bounds on some of the instances and their upper bounds are often far weaker than those identified by ZAMinfluence \cite{broderick2020automatic}. One some runs, their algorithm identifies strong bounds that seemingly contradict the optimal exact solution (e.g., for India); this  reflects the difference in optimization problems, since \cite{moitra2022provably} solves the fractional problem, in which an objective of 8.2 is feasible (as identified by Gurobi, which solves the fractional instance to optimality); ZAMinfluence considers the integral version, for which 9 is the optimal solution (as certified by both Gurobi and our exact algorithm). The spectral algorithm obtains weak lower bounds for only some of these instances.

The next row, Minimum Wage, considers the difference-in-difference setting from Section \ref{ssec:diffs}. Here, we focus on a textbook example of difference-in-difference estimation, specifically \cite{card1993minimum}. We did not implement variants of ZAMinfluence or Moitra and Rohatgi's algorithms for this version of the problem, in which observations have to be dropped in pairs. However, we highlight that Gurobi cannot solve this instance to optimality with a 30-minute time limit, whereas our exact algorithm solves it in less than a second. 

Next, we consider the settings studied in \cite{eubank2022enfranchisement} and \cite{martinez2022much} with Zaminfluence \cite{broderick2020automatic}. Both of these settings are too high-dimensional for our exact algorithms to apply, or those of \cite{moitra2022provably} to converge in reasonable time, yet in both cases Gurobi (and our implementations of ZAMinfluence and the Greedy heuristic) finds significantly smaller subsets than those reported by the respective authors (28 compared to 97 and 2.8\% compared to 5.1\% --- we speculate that the large gap between our influence-based algorithms and those previously used arise from improved numerical stability in our implementations).

Finally, we consider two synthetic datasets.
The Synthetic 2D dataset consists of 100 samples $(X_i,Y_i)$, where $X_i \sim \mathcal{N}(0,1)$ and $Y_i = -2X_i + \epsilon_i$, with $\epsilon_i \sim \mathcal{N}(0,1)$; we consider the regression model $Y = X\beta + \alpha + \epsilon$; i.e. allowing for a fixed-effects/intercept term.
The Synthetic 4D dataset consists of 1000 samples $(X_i,Y_i)$ where $X_i \in \R^4$ has iid coordinates from $\mathcal{N}(0,1)$, and $Y_i = X_i(1) + X_i(2) + X_i(3) + X_i(4) + \epsilon_i$.
We consider the linear model $Y = X(1) \beta_1 + X(2) \beta_2 + X(3) \beta_3 + X(4) \beta_4 + \epsilon$, i.e. without a fixed-effects/intercept term.
The spectral algorithm only produces lower bounds on stability; it is the only algorithm among those we study to produce nontrivial lower bounds for the Synthetic 4D dataset, but performs comparatively poorly on the other datasets.
Gurobi is run with a 60 second cutoff on the synthetic datasets -- 30 seconds allotted to fractional solving, 30 to integer solving.
(Overall runtime is greater than 60 seconds because of the time required for Gurobi to set up the model.)

\begin{table}[h!]
  \begin{center}
    \caption{Algorithmic runtimes in seconds (rounded to the nearest integer and, in most nontrivial cases based on algorithmic parameters). Note that running time for Microcredit studies includes time to solve all 7 studies.}
    \label{tab:table2}
    \begin{tabular}{|l|c|c|r || r|c|c|}
    \hline
    Study/Instance &
    \multicolumn{1}{c|}{\cite{broderick2020automatic}}
    & \multicolumn{1}{c|}{
    \cite{kuschnig2021hidden}}
    & \multicolumn{1}{c||}{\cite{moitra2022provably}}
    & \multicolumn{1}{c|}{Gurobi}
    & \multicolumn{1}{c|}{Exact}
    & {Spectral}\\
    \hline
    Microcredit studies  & 5 & 8 & 3640
     & 
    151 & 0 & 0
     \\ \hline
    Min wage \cite{card1993minimum} & --
     & -- & --
     & 1950 &
     0 & --
     
      \\ \hline
    Incarceration \cite{eubank2022enfranchisement} & 0
     & 1 & --
     &  9 & --
     & 0
      \\ \hline
    GDP \cite{martinez2022much}& 50
     & 122 & --
     & 243 & --
     & 0
      \\ \hline
      Synthetic 2D & 0 
     & 0 & 25
     & 0  & --
     & 0
      \\
      Synthetic 4D & 2 
     & 7 &
    40 (no LB) & 61 & --
     & 0 
      \\ \hline
    \end{tabular}
  \end{center}
\end{table}

\paragraph{Boston Housing Data. }
As discussed above, Moitra and Rohatgi evaluate their algorithms on the well-known Boston housing dataset \cite{moitra2022provably}. For the 156 instances they consider, we display the results (lower and upper bounds) in Figure \ref{fig:bh}. To ensure a fair comparison, we set the parameters affecting the runtime for Gurobi and the \cite{moitra2022provably} algorithms so that they run in approximately the same time; in particular, for all these instances combined Gurobi took about 8 minutes whereas the \cite{moitra2022provably} algorithms took about 12 minutes. 

We first compare the upper bounds, comparing ours with the the ones identified by the Net-algorithm in \cite{moitra2022provably} and the ones identified through ZAMinfluence with resolving \cite{broderick2020automatic,kuschnig2021hidden}, as implemented by \cite{moitra2022provably}. Here we find that the Net algorithm of \cite{moitra2022provably} usually identifies similarly strong upper bounds to Gurobi. Though Gurobi identifies tighter upper bounds on about 99\% of instances, the difference is smaller than 5 on 99\% of instances. In contrast, ZAMinfluence (with or without resolving with either implementation) never identifies a tighter upper bound than Gurobi and is off by at least 20 on about 20\% of instances. Next, in Figure \ref{fig:bh2} we compare the 
lower bounds identified by \cite{moitra2022provably} and by Gurobi, noticing that Gurobi identifies stronger bounds on 93\% of the instances. Finally, in Figure \ref{fig:bh3} we plot the resulting optimality gaps across all instances. This comparison shows that Gurobi obtains tight bounds (within 1\%) on 92\% of the instances, and obtains a lower bound of at least 35\% of its upper bound on all instances. In contrast, \cite{moitra2022provably} does not obtain tight bounds (within 1\%) on 92\% of the instances and obtains a lower bound of at least 20\% of the upper bound on just 85\% of the instances.

\begin{figure}
    \centering
    \begin{subfigure}[t]{\linewidth}
        \centering
        \includegraphics[width=.5\linewidth]{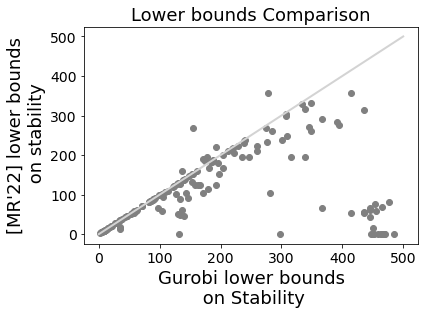} 
        \caption{Comparison of upper bounds} \label{fig:bh1}
    \end{subfigure}
    \begin{subfigure}[t]{\linewidth}
        \centering
 \includegraphics[width=.5\linewidth]{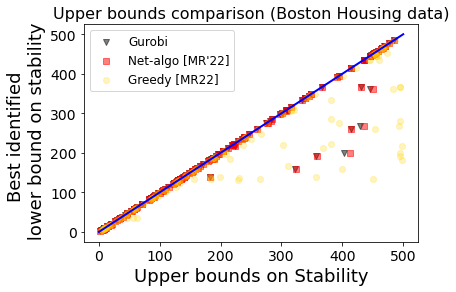} 
        \caption{Comparison of lower bounds} \label{fig:bh2}
    \end{subfigure}
    \begin{subfigure}[t]{.5\linewidth}
        \centering
    \includegraphics[width=1\linewidth]{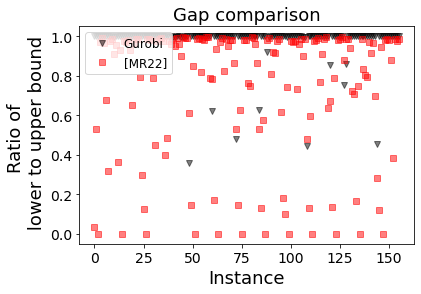} 
        \caption{Comparison of resulting gaps} \label{fig:bh3}
    \end{subfigure}
    \caption{The three plots in this are based on the Boston housing data, as analyzed by \cite{moitra2022provably}. Plot (a) compares the upper bounds obtained in ZAMinfluence with the ones in \cite{moitra2022provably} and ones obtained by Gurobi. Plot (b) compares the lower bounds of the latter two, and Plot (c) compares the resulting optimality gaps.}\label{fig:bh}
\end{figure}

\subsection{Challenge Data}
In our accompanying replication package, we provide \texttt{csv} files for all the datasets above.
Three are designated as \emph{challenge datasets}: Synthetic 4D (\texttt{synthetic4d.csv}), Incarceration (\texttt{Eubank\_black\_perc.csv}), and GDP (\texttt{martinez.csv}).
As described in Table~\ref{tab:table1}, all our methods leave wide gaps between upper and lower bounds on these datasets. In particular, for Incarceration and GDP, we cannot identify any nontrivial lower bounds.
We believe that progress towards closing these gaps requires new algorithms that would constitute substantial steps toward practical  robustness auditing.

\section*{Acknowledgements}
SBH was funded by NSF award no. 2238080 as well as MLA@CSAIL. DF thanks Jacquelyn Pless, Jose Blanchet, Vasilis Syrgkanis, and Rahul Mazumder for insightful conversations.
SBH thanks Nati Srebro for helpful discussions. 

\bibliographystyle{amsalpha}
\bibliography{refs}

\end{document}